\newif\ifnotes
\newif\iflncs
\newif\ifanonymous
\newif\ifexabs

\iflncs
  \documentclass[runningheads]{llncs}
\else
  \documentclass[11pt]{article}
\fi
\iflncs\else
  \usepackage{fullpage}
\fi

\usepackage{iftex}
\ifPDFTeX
  \usepackage[utf8]{inputenc}
  \usepackage[noTeX]{mmap}
  \usepackage[T1]{fontenc}
\fi
\ifLuaTeX
  \usepackage{luatex85}
  \usepackage[noTeX]{mmap}
\fi

\usepackage{multirow}

\usepackage[shortlabels]{enumitem}
\ifexabs
  \setlist[description]{noitemsep}
  \setlist[enumerate]{noitemsep}
  \setlist[itemize]{noitemsep}
\fi
\usepackage{amsmath}
\usepackage{amssymb}
\iflncs\else
  \usepackage{amsthm}
\fi
\usepackage[dvipsnames]{xcolor}
\usepackage{mathtools}
\usepackage[utf8]{inputenc}
\usepackage{amsfonts}
\usepackage{graphicx}
\usepackage{mathrsfs}
\usepackage{braket}
\usepackage{soul}
\usepackage{bm} %
\usepackage[ruled,linesnumbered]{algorithm2e}
\usepackage[T1]{fontenc}
\usepackage[splitrule]{footmisc}
\usepackage[normalem]{ulem} %

\usepackage[style=alphabetic,minalphanames=3,maxalphanames=4,maxnames=99,backref=true]{biblatex}
\renewcommand*{\labelalphaothers}{\textsuperscript{+}}
\renewcommand*{\multicitedelim}{\addcomma\space}
\DeclareFieldFormat{eprint:iacr}{Cryptology ePrint Archive: \href{https://ia.cr/#1}{\texttt{#1}}}
\DeclareFieldFormat{eprint:iacrarchive}{Cryptology ePrint Archive: \href{https://eprint.iacr.org/archive/#1}{\texttt{#1}}}
\AtEveryBibitem{%
 \clearlist{address}
 \clearfield{date}
 \clearfield{isbn}
 \clearfield{issn}
 \clearlist{location}
 \clearfield{month}
 \clearfield{series}

 \ifentrytype{book}{}{%
  \clearlist{publisher}
  \clearname{editor}
 }
}

\usepackage{mleftright}
\newcommand{\paren}[1]{\left(#1\right)}
\newcommand{\mparen}[1]{\mleft(#1\mright)}
\newcommand{\mbracket}[1]{\mleft[#1\mright]}
\newcommand{\braces}[1]{\left\{#1\right\}}
\newcommand{\abs}[1]{\left|#1\right|}
\newcommand{\norm}[1]{\left\lVert#1\right\rVert}
\newcommand{\equad}{\mathrel{\phantom{=}}}

\newcommand{\extrapolateto}{\texorpdfstring{$\rightarrow$}{→}}

\newcommand{\Tr}{\mathrm{Tr}}
\newcommand{\ketbra}[2]{\ket{#1}\!\bra{#2}}
\newcommand{\proj}[1]{\ketbra{#1}{#1}}

\definecolor{DarkBlue}{RGB}{0,0,150}
\definecolor{DarkRed}{RGB}{150,0,0}
\definecolor{DarkGreen}{RGB}{0,150,0}
\usepackage[bookmarks]{hyperref}
\usepackage[nameinlink]{cleveref}
  \crefname{step}{Step}{Steps}

\usepackage{tikz}
\usetikzlibrary{patterns}

\iflncs\else
\newtheorem{theorem}{Theorem}[section]
\newtheorem{proposition}[theorem]{Proposition}
\newtheorem{claim}[theorem]{Claim}
\newtheorem{lemma}[theorem]{Lemma}
\newtheorem{corollary}[theorem]{Corollary}

\newtheorem{remark}[theorem]{Remark}

\newtheorem{definition}[theorem]{Definition}

\fi

\newenvironment{theorem-restated}[1]{
  \renewcommand\thetheorem{\ref{#1}}\theorem}{\endtheorem\addtocounter{theorem}{-1}}
\newenvironment{definition-restated}[1]{
  \renewcommand\thetheorem{\ref{#1}}\definition}{\endtheorem\addtocounter{theorem}{-1}}
\newenvironment{corollary-restated}[1]{
  \renewcommand\thetheorem{\ref{#1}}\corollary}{\endtheorem\addtocounter{theorem}{-1}}
\newenvironment{lemma-restated}[1]{
  \renewcommand\thetheorem{\ref{#1}}\lemma}{\endtheorem\addtocounter{theorem}{-1}}

\newtheorem{fact}[theorem]{Fact}
\newtheorem{construction}[theorem]{Construction}
\Crefname{importedtheorem}{Imported Theorem}{Imported Theorems}
\Crefname{theorem}{Theorem}{Theorems}
\Crefname{proposition}{Proposition}{Propositions}
\Crefname{claim}{Claim}{Claims}
\Crefname{lemma}{Lemma}{Lemmas}
\Crefname{conjecture}{Conjecture}{Conjectures}
\Crefname{corollary}{Corollary}{Corollaries}
\Crefname{construction}{Construction}{Constructions}
\Crefname{property}{Property}{Properties}
\Crefname{game}{Game}{Games}
\Crefname{item}{Item}{Items}
\Crefname{algorithm}{Algorithm}{Algorithms}

\iflncs\else
  \theoremstyle{definition}
\fi

\Crefname{definition}{Definition}{Definitions}
\Crefname{assumption}{Assumption}{Assumptions}
\Crefname{notation}{Notation}{Notations}

\iflncs\else
  \theoremstyle{remark}
\fi

\Crefname{question}{Question}{Questions}
\Crefname{remark}{Remark}{Remarks}
\Crefname{comment}{Comment}{Comments}
\Crefname{fact}{Fact}{Facts}

\def\cA{{\cal A}}
\def\cB{{\cal B}}

\def\cR{{\cal R}}

\providecommand{\RegA}{\mathsf{A}}
\providecommand{\RegB}{\mathsf{B}}
\providecommand{\RegC}{\mathsf{C}}
\providecommand{\RegD}{\mathsf{D}}

\providecommand{\RegK}{\mathsf{K}}

\providecommand{\RegM}{\mathsf{M}}

\providecommand{\RegR}{\mathsf{R}}
\providecommand{\RegS}{\mathsf{S}}

\providecommand{\secp}{\lambda}
\providecommand{\poly}{\mathsf{poly}}

\providecommand{\negl}{\mathsf{negl}}

\providecommand{\Adv}{\mathsf{Adv}}
\providecommand{\com}{\mathsf{com}}

\providecommand{\Com}{\mathsf{Com}}

\providecommand{\D}{\mathsf{D}}
\providecommand{\G}{\mathsf{G}}

\providecommand{\Gen}{\mathsf{Gen}}

\providecommand{\CNOT}{\mathsf{CNOT}}

\providecommand{\TD}{\mathsf{TD}}

\DeclareMathSymbol{\mh}{\mathord}{operators}{`\-}

\providecommand{\Verify}{\mathsf{Verify}}

\DeclareMathOperator*{\E}{\mathbb{E}}

\providecommand{\aux}{\mathsf{aux}}

\providecommand{\Rand}{\mathsf{Rand}}
\providecommand{\CRand}{\mathsf{CRand}}
\providecommand{\init}{\mathsf{init}}

\newcommand{\concat}{\|}

\newcommand{\chal}{\mathsf{chal}}
\newcommand{\Clone}{\mathsf{Clone}}
\renewcommand{\braket}[2]{\left\langle #1 \middle| #2 \right\rangle}

\title{Hard Quantum Extrapolations in Quantum Cryptography}
\iflncs
  \author{Luowen Qian\inst{1}\orcidID{0000-0002-1112-8822} \and Justin Raizes\inst{2} \and Mark Zhandry\inst{1}\orcidID{0000-0001-7071-6272}}
  \institute{NTT Research, Inc. \and Carnegie Mellon University}
\else
    \author{Luowen Qian$^1$ \and Justin Raizes$^2$ \and Mark Zhandry$^1$}
    \date{$^1$NTT Research, Inc. \\ $^2$Carnegie Mellon University}
\fi

\begin{document}
\maketitle

\ifexabs\else
  \begin{abstract}
Although one-way functions are well-established as the minimal primitive for classical cryptography, a minimal primitive for quantum cryptography is still unclear.
Universal extrapolation, first considered by Impagliazzo and Levin (1990), is hard if and only if one-way functions exist.
Towards better understanding minimal assumptions for quantum cryptography,
we study the quantum analogues of the universal extrapolation task.
Specifically, we put forth the classical$\rightarrow$quantum extrapolation task, where we ask to extrapolate the rest of a bipartite pure state given the first register measured in the computational basis.
We then use it as a key component to establish new connections in quantum cryptography: (a) quantum commitments exist if classical$\rightarrow$quantum extrapolation is hard; and (b) classical$\rightarrow$quantum extrapolation is hard if any of the following cryptographic primitives exists: quantum public-key cryptography (such as quantum money and signatures) with a classical public key or 2-message quantum key distribution protocols.

For future work, we further generalize the extrapolation task and propose a fully quantum analogue.
We show that it is hard if quantum commitments exist, and it is easy for quantum polynomial space.
\end{abstract}

  \section{Introduction}
\fi
Modern cryptography works by reducing the security of complicated cryptosystems down to the hardness of solving simpler problems.
These reductions also have the side benefit that it enables us to modularize cryptographic constructions and theorems.
A celebrated example developed through a sequence of such transformations is the fact that one-way functions --- functions that are easy to compute but hard to invert --- are sufficient to realize a large swath of symmetric key cryptography~\cite{FOCS:GolGolMic84,C:LubRac85,STOC:Rompel90,HILL99}. On the other hand, one-way functions are inherent to the vast majority of complexity-based cryptography~\cite{FOCS:ImpLub89,Gol90-ci}. This makes the one-way function abstraction a bedrock of modern cryptography, and cryptographers frequently refer to the assumed existence of one-way functions as the ``minimal'' assumption in cryptography.

The similar question for quantum cryptography is just as important but is a lot less studied for now.
Interestingly, one-way functions are no longer minimal: recent works have demonstrated that using quantum information, cryptography can be based on assumptions that appear strictly milder than even a one-way function~\cite{TQC:kretschmer21,C:AnaQiaYue22,C:MorYam22,KQST23-algorithmica}. This leads to the following fundamental question:
\begin{center}
    \emph{Is there a quantum analog of a one-way function, that is both inherent to essentially all of complexity-based cryptography, while also being sufficient to build useful cryptosystems?}
\end{center}
In recent years, there have been some progress towards answering this question but it does not have a satisfactory answer yet.
To explain this, let us first recall how classical cryptographic security is usually formalized.
Broadly, we can divide them into two categories:
\begin{itemize}
    \item In a \textbf{decision-style} security game, the adversary's goal is to distinguish between two experiments, such as distinguishing whether the encrypted message is $0$ or $1$, or whether the game is in the real world or the ideal world.
    
    Decision-style quantum security assumptions are fairly well understood.
    Specifically, security games of this style can be captured by the notion of an EFI pair, and any such EFI pairs implies a number of cryptographic tasks, such as commitments \cite{AC:Yan22,ITCS:BCQ23}.
    \item In a \textbf{search-style} security game, the adversary's goal is simply to produce certain message to satisfy some predicate, such as forging a signature or a quantum money state.
    
    Search-style quantum security assumptions are a lot less studied.
    One glaring issue is that public-key quantum money, a very natural quantum cryptographic primitive, is not known to be related to the other quantum cryptographic primitives.
    In fact, public-key quantum money and its strengthening, quantum lightning~\cite{EC:Zhandry19b}, are only known from extremely strong assumptions like post-quantum obfuscation or in ideal oracle models, yet we do not even know if any cryptographic assumption such as EFI pairs are necessary for them.
\end{itemize}

A few recent works have studied what we call a \textbf{quantum$\rightarrow$classical} search-style assumption, where the adversary is given a potentially quantum input, and must produce some classical output.
An example is a quantum one-way state generator \cite{TQC:MorYam24}, where informally a keyed mixed state can be efficiently prepared from its classical key but it is hard to find its key given the state.
It was shown that quantum one-way state generators with inefficient verification algorithms are existentially equivalent to EFI pairs~\cite{STOC:KT24,FOCS:BJ24}.
As a consequence of this, EFI pairs are implied by primitives like secret-key quantum money, given the scheme has a classical secret.
However, public-key quantum money may not have a classical secret, and quantum lightning must not have a classical secret.
Thus, this does not give us a way to construct EFI pairs from public-key quantum money.

Going to more general search-style assumptions, very little is known about the search-style assumption where the adversary is supposed to output some quantum state, or even a \textbf{classical$\rightarrow$quantum} search-style assumption such as public-key quantum money.
To make progress towards addressing the fundamental question above, we can ask a more concrete question:
\begin{center}
    \emph{Is it possible to build EFI pairs from public-key quantum money?}
\end{center}

\subsection{Classical\extrapolateto{}Quantum Extrapolation}

To understand how EFI pairs, or equivalently, commitments may be built from public-key cryptography, let us take a step back and think about how a classical analogue of this is proved, for example, how we can construct a classical commitment scheme or a one-way function from a classical signature scheme.
The key middle step of the classical proof \cite{FOCS:ImpLev90} is the hardness of the \emph{universal extrapolation} task.
In particular, it is shown that secure signature schemes imply hardness of universal extrapolation, which in turn implies the existence of (distributional) one-way functions.

Inspired by this template, we study quantum analogues of this extrapolation task, which we emphasize is a search-style assumption.
We first define classical$\rightarrow$quantum extrapolation which, given an efficient quantum pure state, asks to extrapolate the rest of the quantum state conditioned/post-selected on measuring the first half of the state (in the computational basis).

\begin{definition}[Classical$\rightarrow$Quantum Extrapolation, informal]
    A classical$\rightarrow$quantum extrapolation problem is specified by a circuit $\Gen$ that produces a pure state, which can be written as
    \[
        \ket{\Gen} = \sum_{s} \alpha_s \ket{s}_{\RegA} \otimes \ket{\psi_s}_{\RegB}
    \]
    for some $\alpha_s \ge 0$ and unit vectors $\ket{\psi_s}$.
    We say (a uniform family of) $\Gen$ is hard if for every QPT adversary $\Adv$ (potentially with auxiliary input), its output has a negligible overlap with the correct state $\ket{\psi_s}$ given the classical part $s$, i.e.
    \[
        \E[\Tr(\proj{\psi_s}\Adv(s))] = \Tr\mparen{\sum_s \alpha_s^2 \proj{\psi_s} \Adv(s)}
    \]
    is negligible.
\end{definition}

This is a natural generalization of the classical extrapolation task for distributions to pure states, but instead of asking for the conditional distribution on $\RegB$ conditioned on $s$, we use a stronger requirement of preparing the pure (unmeasured) state $\ket{\psi_s}$.

Using similar proof ideas as classically, we can straightforwardly establish the following.

\begin{proposition}
    There exists a hard classical$\rightarrow$quantum extrapolation task if any of the following exists:
    \begin{itemize}
        \item Public-key quantum money scheme (with a classical public key).
        \item Public-key quantum signature scheme (with a classical public key).
        \item 2-message quantum key distribution that is only unpredictably-secure.
    \end{itemize}
\end{proposition}

Our main theorem then finishes the proof by establishing the following.

\begin{theorem}
    If there exists a hard classical$\rightarrow$quantum extrapolation task, then quantum commitment schemes exist.
\end{theorem}

In fact, our commitment construction trivially adapts to \emph{cloneable}$\rightarrow$quantum extrapolation, where we only require the challenge $\RegA$ register to be cloneable instead of strictly classical.
As a consequence, we can construct commitments from public-key quantum money scheme with a cloneable quantum public key as well.
Interested readers should refer to the formal treatment later.

We remark that these results are optimal in a few perspectives.
First, asking for the public key to be cloneable is barely a restriction: if one would like to reap the benefits of having a public key infrastructure, it is important for the authority to be able to efficiently clone and distribute public keys without all the users being online to produce new copies of the public key.
Second, it is known that uncloneable public-key signatures (with bounded number of copies of the public key) are in fact statistically possible~\cite{arxiv:GC01}.
Third, 3-message quantum key distribution with unpredictability security is also statistically possible.

\ifexabs
\else
\begin{remark}
    It might be tempting to think that some form of distributional one-wayness should follow from hard classical$\rightarrow$quantum extrapolation tasks.
    Classically, consider a sampler of a hard-to-extrapolate distribution $Samp(r) \rightarrow (a, b)$.
    Then the truncated sampler $f(r) \rightarrow a$ is immediately a distributionally one-way function.
    However, since a quantum algorithm can build a pure state from scratch, it is unclear if any meaningful form of distributional one-wayness can be constructed here.
    From this perspective, it appears that extrapolation may be a more useful abstraction than distributional one-wayness.
\end{remark}
\fi

\subsection{Quantum Extrapolation}

Given the result above, it is natural to wonder if some version of the following generalization is true: any search-style assumption implies EFI pairs, even if both the inputs and outputs are potentially quantum.
With that said, to the best of our knowledge, all natural examples (binding security of commitments, or soundness of zero knowledge arguments) already imply EFI pairs.

Hence, we put forth a candidate \emph{quantum input, quantum output} search task.
Intuitively, the task is to convert a bipartite pure state into its canonical purification.
We show that this task is implied by commitment schemes and on the other hand, can be solved in quantum polynomial space.
We do not know how to use it to construct other primitives such as commitments, and we provide a discussion of the difficulties in extending our construction in \ifexabs\else\Cref{sec:tech-q2qe}\fi.

\begin{theorem}[Informal]
    \label{thm:informal-qextrapolation}
    If quantum bit commitments exist, then there exists an efficiently preparable state with Schmidt decomposition
    \[
        \sum_{i} \alpha_i \ket{A_i} \otimes \ket{B_i},
    \]
    such that it is hard to coherently ``strongly'' map $\ket{B_i}$ to $\ket{A_i^*}$ on average over $\alpha_i$.
    On the other hand, this task is possible in $\mathsf{pureUnitaryPSPACE}$ for any efficiently preparable state.
\end{theorem}

\paragraph{Discussion: quantum minimalism.}
Mark Zhandry's talk on ``Quantum Minimalism'' \cite{Zha23-minimalism} highlighted various features that make one-way functions ``minimal'' for classical cryptography.
We recall these features here and highlight the fact that quantum extrapolation satisfy the first 7 proposed features to roughly the same extent as classically.
Thus, hard quantum extrapolations might be a good candidate for the minimal quantum cryptographic assumption from this perspective if we can build useful cryptography from these.
\begin{description}
    \item[Feature 1: (Somewhat) trivially implied by most general primitives] As discussed above, this matches one-way functions since we can somewhat trivially prove hardness of quantum extrapolation from commitments, which is implied by most general primitives.
    \item[Feature 2: Trivially and robustly implied by most concrete assumptions] This holds for all concrete assumptions that we can think of.
    \item[Feature 3: Simple to define] Modulo the complex conjugate condition which is necessary in order to satisfy other features, this is essentially true.
    \item[Feature 4: Falsifiable] This is probably not true in general. However, in its defense, its classical analogues, universal extrapolation hardness or distributional one-way functions, are not known to be falsifiable either. (If falsifiable assumptions like commitments are implied by this, then this feature is not as essential.)
    \item[Feature 5: Search problem] Quantum extrapolation is defined as a search problem.
    \item[Feature 6: Trivial combiners and universal constructions] A trivial combiner can be constructed by simply concatenating the respective parts of the two states: this thus implies a universal construction via standard techniques.
    \item[Feature 7: Minimal correctness requirements] This holds since the only requirement on the state is that it should be efficiently preparable.
    \item[Feature 8: Can build crypto?] We leave as interesting open question whether the hardness of quantum extrapolation can be used to construct EFI pairs or other cryptographically useful hardness.
\end{description}

\ifexabs\else
\section{Technical Overview}

\subsection{Why a New Approach?}

Before we discuss our commitment construction, one might wonder if it might be simpler to generalize the construction of commitments (or EFI pairs) from quantum one-way state generators \cite{STOC:KT24,FOCS:BJ24} to this setting.
After all, their techniques can handle quantum$\rightarrow$classical search-style assumption and we only need to do the other way around.

Upon closer inspection, it appears that unfortunately their techniques crucially rely on having access to a classical secret.
More specifically, their construction of EFI pair on a high level still goes through the classical construction of hardcore predicates from one-way functions, and this in turn builds on the fact that the secret is classical.

To see this, recall that the construction essentially leaks some random inner product about the secret and argues that such a leakage appears computationally random to the adversary.
However, the presence of a quantum secret shuts down this approach due to the lack of a quantum analogue of the Goldreich--Levin hardcore bit.
In fact, some natural formulation of such an analogue is flat out impossible.
For example, assume the quantum secret is a (computationally) Haar random state. Then a non-trivial leakage would instead be statistically independent of the secret, whereas the GL hardcore bit is statistically determined by the secret.

Another na\"ive approach is to try to apply the construction of one-way functions from hardness of universal extrapolation to the classical$\rightarrow$quantum extrapolation setting.
Unfortunately, this construction also applies randomness extractors on the secret and we run into a similar difficulty as above.

\subsection{Commitments from Classical\extrapolateto{}Quantum Extrapolation}

Given the known approaches of constructing EFI pairs do not seem to work, it appears that a drastically different approach is needed.
The initial idea is that since the binding security game of the commitment scheme is already a search-style assumption and is known to be existentially equivalent to EFI pairs, it might be easier to construct a computationally \emph{binding} commitment instead.
Perhaps we can carefully craft a commitment scheme so that breaking its binding would directly correspond to solving the extrapolation problem.
This idea was also used to construct commitments from secretly-verifiable statistically-invertible one-way state generators \cite{TQC:MorYam24}.

To construct a commitment, we will use the canonical commitment scheme template \cite{AC:Yan22}.
As a starting point, imagine a sender who manufactures $\ket{\Gen}$ in the opening register $\RegD$ and copies $s$ to the commitment register $\RegC$. This results in the state
\[
    \sum_s \alpha_s \ket{s}_\RegC \otimes \ket{s, \psi_s}_\RegD,
\]
which we will use as the commitment and the decommitment register if we want to commit to bit $1$.

Now imagine if we could also ``obliviously'' quantum sample the challenge distribution, that is, we can prepare the following pure state
\[
    \sum_s \alpha_s \ket{s}_\RegC \otimes \ket{s}_\RegD
\]
for the case of committing to $0$.
For example, if the reduced density matrix on $s$ is maximally mixed, then this is essentially asking to prepare the maximally entangled state, which can be efficiently done.

However, note that if this were indeed possible, we would be done as the commitment would be computationally binding.
If a malicious committer commits to $0$ and wishes to change his mind to decommit to $1$, he would essentially be forced to craft the state $\ket{\psi_s}$ given only $s$.%

\paragraph{Removing quantum sampling assumption.}
Unfortunately, not every efficient distribution can be quantum sampled unless, for example, $\mathsf{SZK} \subseteq \mathsf{BQP}$ \cite{AT07-qsample}.
Thus we want to also handle the case where $s$ can be arbitrarily distributed, which could be the case for a general quantum money scheme.

A first idea is that maybe starting from the $0$-commitment state, we can ``effectively'' remove $\ket{\psi_s}$ from the committer's view (or equivalently, its register $\RegD$).
Na\"ively we could simply assign the state $\ket{\psi_s}$ to the $\RegC$ register.
However, the resulting commitment scheme would be not hiding, so some care has to be taken.

The next idea, then, is to maybe employ some encryption scheme to do this.
More specifically, our second attempt is as follows:
\begin{align*}
    \ket{\com_0} &= \E_k\mbracket{\sum_s \alpha_s \ket{s, k}_\RegC \otimes \ket{s, Enc(k, \psi_s)}_\RegD}, \\
    \ket{\com_1} &= \E_k\mbracket{\sum_s \alpha_s \ket{s, k}_\RegC \otimes \ket{s, Enc(k, 0), \psi_s}_\RegD},
\end{align*}
for some (unitary) encryption scheme $Enc$.

Now, what encryption scheme would work?
The most natural candidate, quantum one-time pad, turns out not to work: if $\ket{\psi_s}$ are all encoded in Hadamard basis then the resulting scheme would again not be hiding.
More specifically, consider the state $1 \otimes \ket{+}^{\otimes n}$.
Then (ignoring the $\ket{+}^{\otimes n}$ part in $\com_1$ which is irrelevant for hiding)
\begin{align*}
    \ket{\com_0} &= \E_{k \in \{0, 1\}^{2n}}\mbracket{\ket{k}_\RegC \otimes \ket{Enc(k, \ket{+^n})}_\RegD} = \mparen{\ket{+}_{\RegC_X} \otimes \frac{\ket{0+}_{\RegC_Z\RegD} + \ket{1-}_{\RegC_Z\RegD}}{\sqrt2}}^{\otimes n}, \\
    \ket{\com_1} &= \E_{k \in \{0, 1\}^{2n}}\mbracket{\ket{k}_\RegC \otimes \ket{Enc(k, 0)}_\RegD} = \mparen{\ket{+}_{\RegC_Z} \otimes \frac{\ket{00}_{\RegC_X\RegD} + \ket{11}_{\RegC_X\RegD}}{\sqrt2}}^{\otimes n}.
\end{align*}
Thus we see that these two density matrices have statistical distance close to $1$ as $n \to \infty$.

Nevertheless, we make the crucial observation: this attack appears to rely on the secret quantum state being encoded in a specific basis.
If we attempt to encrypt, instead of only in two bases, in exponentially many bases, then the probability that we hit such a bad basis would be negligible.
To make this efficient, let us say that we just pick one of bases to encrypt uniformly at random.
Since adding random phases to a basis is equivalent to measuring it, this gives us the following random-measurement commitment scheme.

\begin{construction}[Informal]
    Consider an (exponentially large) family of complete measurement bases $M = \{\ket{\phi_{k, x}}\}$ indexed by $k$ and measurement outcome $x$.
    Given a quantum state $\ket\psi$, we use the notation $\ket{M_k(\psi)} \otimes \ket{M_k(\psi)}$ to informally denote (coherently) measuring $\ket\psi$ in basis $k$ and then cloning the outcome into the second register.
    More formally, we use this notation to denote the following pure state
    \[
        \ket{M_k(\psi)} \otimes \ket{M_k(\psi)} := \sum_x \braket{\phi_{k, x}}{\psi} \ket{x} \otimes \ket{x}.
    \]
    Then the random-measurement commitment with bases $M$ is the following construction:
    \begin{align*}
        \ket{\com_0} &= \E_k\mbracket{\sum_s \alpha_s \ket{s, k, M_k(\psi_s)}_\RegC \otimes \ket{s, k, M_k(\psi_s)}_\RegD}, \\
        \ket{\com_1} &= \E_k\mbracket{\sum_s \alpha_s \ket{s, k, M_k(0)}_\RegC \otimes \ket{s, k, M_k(0), \psi_s}_\RegD}.
    \end{align*}
\end{construction}

In the encryption perspective, this is equivalently using a random encryption scheme indexed by some public randomness $k$, and secret random phases are added to the basis specified by $k$.
Moving forward, we will use the measurement perspective since it is more convenient for the analysis.

We now need to argue that the commitment is hiding and binding and neither is obvious at this point.
Intuitively, hiding should hold as our intuition above indicates that a sufficiently random set of bases should evade the bad distinguishing attack with overwhelming probability.
On the other hand, giving the adversary a complete measurement of the state in a probably useless basis should not help the adversary to swap the state out with the zero state --- or in the encryption perspective, the encryption should be secure enough that the adversary could only build the state from scratch instead of extracting useful information from the encryption.

\paragraph{Weak statistical hiding.}
As observed above, the family of measurement bases must be sufficiently large, and different bases need to be ``independent'' enough to avoid the attack above.
Given these criteria, we consider a few candidates: (1) mutually unbiased bases (MUBs); (2) unitary $t$-designs; (3) stripped down $t$-designs such as binary phase bases \cite{MPSY24-pru}.
It turns out that both MUBs and $t$-designs for $t \ge 2$ work, while $1$-designs and binary phase bases do not necessarily work.
For comparison, MUBs have very good randomness complexity ($n$ bits of randomness for hiding $n$ qubits) while $2$-designs or random Cliffords are more efficient (can be applied in quasi-linear time).
For interested readers, we explain the counterexamples in \iflncs{}the full version\else\Cref{appendix:alternative-bases}\fi.

On a high level, the hiding proof can be reduced the following question: given any two mixed states $\rho_0, \rho_1$ that could be statistically far; is it true that the states after measuring them in a random basis $k$, $(\rho_0^{(k)}, k), (\rho_1^{(k)}, k)$ are somewhat statistically close?
This also might be an independently interesting question on its own beyond this application: traditionally, tomography asks for a set of measurements that preserves the quantum information; here, we ask for a set of measurements that destroys quantum information with respect to statistical distance.

Crucially, we need to bound the statistical distance \emph{even if the basis choice $k$ is revealed}.
This is problematic for the proof of $t$-designs, since $t$-designs generally do not guarantee anything if the secret key $k$ is revealed: everyone can check that the specific unitary corresponding to $k$ is applied instead of a Haar random unitary.
However, there is a simple trick we can apply: the statistical distance we need to bound above can be equivalently expressed as $\E_k[TD(\rho_0^{(k)}, \rho_1^{(k)}]$, as the distribution over $k$ is identical.
Furthermore, here the key is no longer given to the distinguisher.
Even though trace distance is not a polynomial (due to the absolute value), we can nevertheless bound it by an appropriate degree-$2$ polynomial and then invoke the security of $2$-designs.
In the end, we show that the average statistical distance under $2$-designs is upper bounded by $\norm{\rho_0 - \rho_1}_2/2$.
A na\"ive upper bound would be $\norm{\rho_0 - \rho_1}_2/2 \le \norm{\rho_0 - \rho_1}_1/2 = \TD(\rho_0, \rho_1)$ which yields a trivial bound.
However, by more carefully leveraging Jordan--Hahn decomposition of $\rho_0 - \rho_1$, we are able to show that $\norm{\rho_0 - \rho_1}_2/2 \le \TD(\rho_0, \rho_1)/\sqrt2$.
Curiously, it appears that even measurements under Haar random unitaries would still leave behind $1/2$ statistical distance.

To prove MUBs work, we leverage a useful fact by Ivanovic \cite{Ivonovic1981,Ivanovic92} that any density matrix can be decomposed into the sum of its post-measurement states over all MUBs; furthermore, each pair of states are orthogonal with respect to the Hilbert--Schmidt inner product after appropriate renormalization.
Utilizing the orthogonality, we can translate the average ($L^1$) statistical distance to ($L^2$) Hilbert--Schmidt distance to invoke the decomposition, and then back to statistical distance using the same trick above without losing too much.
In the end, we can compute that the average statistical distance under MUBs is at most $1/\sqrt2$ as well.

We refer interested readers to the later sections for a more detailed proof.

\paragraph{Computational binding.} To show (honest) binding, we show that any adversary who is given register $\RegD$ of $\ket{\com_0}$ and successfully opens to $1$ can be used to solve the classical$\rightarrow$quantum extrapolation problem.
In the classical$\rightarrow$quantum extrapolation game, the reduction is essentially given a mixed state $\sum_{s} \alpha_s^2 \proj{s}$.
In order to meaningfully invoke the binding adversary, the reduction must produce the reduced density matrix on $\RegD$ of $\com_0$, which looks like
\[
    \sum_{k, s} \alpha_s^2 \proj{s}\otimes \proj{k} \otimes \rho_{k,s} \otimes \proj{\vec{0}}
\]
when $\RegC$ is traced out, where $\rho_{k,s}$ denotes the mixed state resulting from measuring $\ket{\psi_s}$ in the basis $k$.
Observe that since the canonical-form commitment receiver is going to project on to the $\ket{\com_1}$ state, the adversary's output must (at the very least) be the following form:
\[
    \sum_{k, s} \alpha_s^2 \proj{s}\otimes \proj{k} \otimes \proj{\psi_s} \otimes \rho'_{k,s}
\]
for some $\rho'_{k,s}$.
This is true since the committer already holds a copy of $k$ and $s$.
Therefore, as long as we can mimic the distribution on $k$ and $\rho_{k, s}$, the reduction would work.

Recall our intuition above that $\rho_{k,s}$ should disclose very little information so that it would not help the adversary to prepare $\ket{\psi_s}$.
This suggests that maybe the reduction can just put there maybe a maximally mixed state and hope that the adversary would not notice.
However, recall from above that $\rho_{k,s}$ is only \emph{weakly} hiding.
In other words, it may contain some information about $\ket{\psi_{s}}$. It is not clear that the reduction can generate such a state without already violating the hardness of the extrapolation problem. 

Our solution is to simply generate $\rho_{k,0}$, a measurement of $\ket{\vec{0}}$ in the basis $k$, and feed it to the binding adversary instead.
In other words, we are again measuring all zeroes and expect the adversary to not see a difference.
This turns out to work.
Although $\rho_{k,0}$ only matches $\rho_{k,s}$ to a limited degree, it \emph{does} match register $\RegC$ from $\ket{\com_0}$. Since register $\RegC$ in $\ket{\com_1}$ would contain $\rho_{k,s}$ instead of $\rho_{k,0}$, the support states where $\rho_{k,s}$ significantly differs from $\rho_{k,0}$ precisely cover the case where the commitment is \emph{statistically binding}. Thus, the assumed binding adversary's advantage must be entirely on the portion of $\rho_{k,s}$ which matches $\rho_{k,0}$.
In order for the substitution to work, the receiver's view of $\rho_{k,0}$ in $\ket{\com_0}$ must be symmetric with the sender's view of $\rho_{k,s}$ in $\ket{\com_1}$. Therefore if the choice of basis $k$ appears in one view, it must appear in both.

\subsection{Quantum Extrapolation}
\label{sec:tech-q2qe}

Here, we discuss the more general task of quantum extrapolation, where an adversary operates on register $\RegB$ of the Schmidt decomposed state
\[
    \sum_{i} \alpha_i \ket{A_i}_{\RegA} \otimes \ket{B_i}_{\RegB}.
\]
For the sake of exposition, we omit conjugations of the states and assume it is hard to coherently map $\ket{B_i}$ to $\ket{A_i}$ on average over $\alpha_i$. Technically, the conjugation is necessary for ensuring that the solution is well-defined, since the Schmidt decomposition of a state is not necessarily unique. See the technical sections for a formal treatment.

\paragraph{Implications from Public-Key Quantum Money and Commitments.}
It is not hard to see that both public-key quantum money and commitments imply hard quantum extrapolation tasks. 
In the case of public-key quantum money, we may regard the serial number as $\ket{B_i}$ and the banknote as $\ket{A_i}$; thus, any adversary mapping $\ket{B_i} \mapsto \ket{A_i}$ counterfeits banknotes using their serial numbers.

The case of commitments is only slightly more complicated. A perfectly hiding commitment in canonical quantum form has a Schmidt decomposition
\begin{align*}
    \ket{\com_0} &= \sum_{i} \alpha_i \ket{c_i}_{\RegC} \otimes \ket{d_{i,0}}_{\RegD},
    \\
    \ket{\com_1} &= \sum_{i} \alpha_i \ket{c_i}_{\RegC} \otimes \ket{d_{i,1}}_{\RegD}.
\end{align*}
Note that because the commitment is perfectly hiding, the left hand side of the two decompositions are the same. Thus, if an adversary could coherently map $\ket{d_{i,0}} \mapsto \ket{c_i} \mapsto \ket{d_{i,1}}$, they could open $\ket{\com_0}$ to $\ket{\com_1}$ by operating only on the opening register, breaking computational binding.

Handling statistical hiding is essentially the same reduction, but the analysis is more difficult.
Essentially, we need to argue that statistical hiding implies that the two target states are close.
We prove this by establishing that the closeness of the target states is captured by the Holevo fidelity and then utilizing a tight connection between Holevo fidelity and trace distance \cite{Kho72-fidelity}.

\paragraph{Difficulties in Building Commitments.} 
It is natural to wonder whether our construction from cloneable$\rightarrow$quantum extrapolation can be extended to the more general quantum extrapolation task.
A natural attempt is to prepare the extrapolation state in the opening register $\RegD$, then ``destroy'' either $\ket{A_i}$ or $\ket{B_i}$, depending on the message bit, by measuring it and copying it in a random basis to the commitment register $\RegC$. This results in a commitment which looks like the following:
\begin{align*}
    \ket{\com_0} &= \E_k\mbracket{\sum_{i} \alpha_i \Ket{k, M_{k}\mparen{A_i}}_\RegC \otimes \Ket{k, B_i, M_{k}\mparen{A_i}}_\RegD},
    \\
    \ket{\com_1} &= \E_k\mbracket{\sum_{i} \alpha_i \Ket{k, M_{k}\mparen{B_i}}_\RegC \otimes \Ket{k, A_i, M_{k}\mparen{B_i}}_\RegD}.
\end{align*}

In order to open $\ket{\com_1}$ to $\ket{\com_0}$, it would suffice to map $\ket{B_i} \mapsto \ket{A_i}$ and let the statistical hiding take care of the differences between $M_{k}(A_i)$ and $M_{k}(B_i)$. Unfortunately, it is not clear whether this mapping is \emph{necessary} to break binding. In particular, the commitment register only contains partial information about $\ket{A_i}$, so the adversary is not necessarily bound to any particular $i$ when breaking binding.
instead, it could potentially map $\ket{B_i}$ onto a superposition of $\ket{A_j}$'s.

For a more concrete example, consider the state $\sum_x \ket x \otimes H\ket x = \sum_{xy} (-1)^{x \cdot y} \ket x \otimes \ket y$ and the fixed measurement of measuring in the computational basis (note that our binding proof before works for any family of measurement bases).
Then the resulting commitment has identical commitment states: $\sum_{xy} (-1)^{x \cdot y} \ket x \otimes \ket{y, x} = \sum_{xy} (-1)^{x \cdot y} \ket y \otimes \ket{x, y}$.
This means that the action to break binding is the identity unitary, whereas we would expect the adversary to apply the Hadamard gates.

\subsection{Related Works}

Both our construction and the commitment construction from one-way state generator \cite{STOC:KT24,FOCS:BJ24} (can) make use of unitary $t$-designs.
Interestingly, the similarity is superficial since the purposes are opposite.
(As a historical note, we actually arrived at this construction before we became aware of \cite{STOC:KT24}.)
In our work, we use $2$-designs to reduce statistical distance, whereas in their work, $3$-designs are indirectly used by classical shadow to preserve information about the quantum state since they are tomographically complete.

A\ifanonymous\else{}fter we have announced our results, a\fi{} more recent work by Dakshita Khurana and Kabir Tomer \cite{KT24} independently constructs commitment schemes from a state puzzle, which is equivalent to our definition of a hard classical$\rightarrow$quantum extrapolation task.
They also prove an amplification theorem for state puzzles, thus commitments can also be built from weak state puzzles, or equivalently, weakly-hard classical$\rightarrow$quantum extrapolation tasks.
However, their techniques would not immediately extend if we instead have a hard cloneable$\rightarrow$quantum extrapolation task.
To see this, their main idea is to view this hardness as a hard state synthesis problem, and construct hard one-way puzzles (or classical$\rightarrow$classical extrapolation tasks) from this hardness using similar techniques as for solving state synthesis with a classical oracle; afterwards, commitments are built using prior work \cite{STOC:KT24}.
In comparison, we construct a commitment directly from hard classical$\rightarrow$quantum extrapolation tasks and the construction trivially extends to cloneable bases as well.

\section{Commitments}

We recall the definition of a canonical non-interactive quantum bit commitment from~\cite{AC:Yan22}.
\begin{definition}[Quantum Bit Commitment Syntax]
    A canonical (non-interactive) quantum bit commitment) is specified by a family of unitaries $\{\Com_\secp\}_{\secp\in \mathbb{N}}$ which acts on two registers $\RegC$ and $\RegD$. It consists of two stages:
    \begin{itemize}
        \item \textbf{Commit.} To commit to a bit $b$, the sender prepares the state $\ket{b} \otimes \ket{\vec{0}}$ in register $(\RegC, \RegD)$, then applies $\Com_{\secp}$ to it to obtain the state $\ket{\com_b}_{\RegC,\RegD}$.\footnote{To simplify the notation, we write $\ket{\vec{0}}$ to denote enough $\ket{0}$ states to finish filling the register.} It sends register $\RegC$ to the receiver as the commitment register and keeps $\RegD$ as the opening register.
        \item \textbf{Open.} To open the commitment, the sender sends register $\RegD$ to the receiver. The receiver applies $\Com_{\secp}^{\dagger}$ to register $(\RegC, \RegD)$, then measures the register in the computational basis. If the measurement result is of the form $b \concat \vec{0}$ for a bit $b$, the receiver outputs $b$. Otherwise they output $\bot$.
    \end{itemize}
\end{definition}

Note that a canonical commitment inherently enjoys completeness; the receiver will always output $b$ as the result of opening a commitment to $b$, because its measurement for output $b$ is exactly a projection onto the state the sender prepares.

A quantum bit commitment must also satisfy hiding and a notion of binding.
We consider honest-binding, which informally guarantees that no adversary given an honestly-prepared commitment to $0$ can open it to $1$ instead, and vice-versa.
It is known that honest-binding for canonical form commitments suffices for stronger binding security \cite{AC:Yan22}.

\begin{definition}[Honest Binding]
    A commitment scheme $\{\Com_\secp\}_{\secp\in \mathbb{N}}$ is computationally (resp. statistically) $\epsilon$-\emph{honest-binding} if for every sufficiently large security parameter $\secp$, every auxiliary input $\ket{\psi}$ in register $\RegA$, and every polynomial-time (resp. physically) realizable unitary $U$ operating on register $(\RegA, \RegD)$,
    \begin{equation*}
        \Tr[\left(I\otimes \proj{\com_{1}} \right) U \left(\proj{\psi}\otimes \proj{\com_0}\right) U^\dagger]\leq \epsilon
    \end{equation*}
    When $\epsilon = \negl$, we simply refer to the commitment as honest-binding.
\end{definition}

It is known that binding on $0 \to 1$ as defined above also implies binding on $1 \to 0$, where the adversary's task is instead transforming an honest commitment to the $1$ bit into one to the $0$ bit.

\begin{definition}[Commitment Hiding]
    A commitment is computationally (resp. statistically) $\epsilon$-hiding if the states
    \begin{align*}
        &\Tr_{\RegD}\left(\proj{\com_0} \right) 
        \\
        &\Tr_{\RegD}\left(\proj{\com_1} \right)
    \end{align*}
    are $\epsilon$-computationally (resp. $\epsilon$-statistically) indistinguishable for every sufficiently large $\secp$. When $\epsilon = \negl$, we simply refer to the commitment as hiding.
\end{definition}

While the security presented here may appear weak, it is known that canonical form commitments satisfying honest binding and hiding is sufficient for constructing commitments with stronger security.

\section{Statistical Hiding via Random Measurements}\label{sec:randomizing-bases}

We give two methods to coherently ``destroy'' the information in a quantum state by dividing it into two registers so that each register on its own contains only a limited amount of information about the original state. Concretely, consider measuring a $n$-qubit register $\cR$ with respect to a basis $B_i$ which is drawn from a set of bases $\cB$, then outputting both the measurement result and the choice of basis $i$. We prove that for certain sets of bases, the expected distance between the distributions induced by this procedure on any two states (averaged over the choice of basis, which equals the total trace distance) is bounded away from $1$.

We first prove hiding when the set of bases $\cB$ is mutually unbiased.
To do this, we recall the definition of mutually unbiased bases and their properties.

Consider an $N$-dimensional quantum state.
A complete measurement is a set of rank-1 projectors $\{P_i\}_{i\in [N]}$  such that $\Tr\mparen{P_i P_j} = \delta_{ij}$.
Two complete measurements $\{P_i\}_{i\in [N]}$ and $\{Q_i\}_{i\in [N]}$ are mutually unbiased if $\Tr\mparen{P_i Q_j} = 1/N$ for all $i, j$.
A maximum set of mutually unbiased bases (MUBs) is a set of $(N + 1)$ complete measurements $\{\{P_i^{(r)}\}_{i\in [N]}\}_{r\in [N+1]}$ that are mutually unbiased.
For the qubit-string case ($N = 2^n$), it is known how to measure any complete measurement in a maximum set of MUBs in time $O(n^3)$ \cite[Section 5.1]{DPS04}.

All Hermitian $N$-dimensional matrices form an $N^2$-dimensional real vector space under the Hilbert--Schmidt inner product $\langle X, Y\rangle := \Tr\mparen{XY}$.
This also induces the Hilbert--Schmidt norm $\norm{X} = \sqrt{\Tr\mparen{X^2}}$.

\begin{lemma}[{\cite{Ivonovic1981,Ivanovic92}}]
    When a maximum set of MUBs is known, a density matrix $W$ has the following orthogonal decomposition
    \begin{equation}
        \label{eq:mub-decomposition}
        W = \frac{I}N + \sum_r\paren{W^{(r)} - \frac IN}
    \end{equation}
    with respect to the Hilbert--Schmidt inner product, where $W^{(r)} := \sum_i P_i^{(r)}WP_i^{(r)}$ is the projection onto the $r$-th MUB.
    Furthermore, each term is mutually orthogonal even when considering two different density matrices.
\end{lemma}

Leveraging these, we prove the following lemma, which intuitively states that the statistical distance would be somewhat hidden under most MUB measurements.
Notably, this is true even if the distinguisher knows the measurement basis.

\begin{lemma}\label{lem:mub}
    Let $N$ be a prime power and let $\cB$ be any set of $p \le N+1$ MUBs. For any two density matrices $W_0$ and $W_1$, the expected trace distance between the post-measurement state of $W_0$ and $W_1$ under a random $r\gets \cB$ is
    \[
        \E_{r\gets \cB}\mbracket{ \TD\mparen{W_0^{(r)}, W_1^{(r)}}}
        \le
        \sqrt{\frac N{2p}} \cdot \TD\mparen{W_0, W_1}.
    \]
\end{lemma}

In particular, if we instantiate $p = N$, then we get statistical distance to be at most $1/\sqrt2$.

\begin{proof}
Leveraging the orthogonal decomposition from \eqref{eq:mub-decomposition}, we can see that
\begin{align*}
    \norm{W_0 - W_1}^2
    &= \norm{\paren{\frac IN - \frac IN} + \sum_{r \in [N + 1]}\paren{\paren{W_0^{(r)} - \frac IN} - \paren{W_1^{(r)} - \frac IN}}}^2 \\
    &= \sum_{r \in [N + 1]}\norm{\mparen{W_0^{(r)} - \frac IN} - \mparen{W_1^{(r)} - \frac IN}}^2 \\
    &= \sum_{r \in [N + 1]}\norm{W_0^{(r)} - W_1^{(r)}}^2 \\
    &\ge \sum_{r \in \cB}\norm{W_0^{(r)} - W_1^{(r)}}^2,
\end{align*}
where the second equality follows from the fact that all the cross terms are $0$.
Dividing both sides by $p$, we get
\begin{align*}
    \E_r\mbracket{\norm{W_0^{(r)} - W_1^{(r)}}^2} &\le \frac{1}{p}\norm{W_0 - W_1}^2\\
    &= \frac1{p}\norm{W_+ - W_-}^2 \\
    &= \frac{\norm{W_+}^2 + \norm{W_-}^2}{p} \\
    &\le \frac{\norm{W_+}_1^2 + \norm{W_-}_1^2}{p} \\
    &= \frac{\norm{W_0 - W_1}_1^2}{2p},
\end{align*}
where the second equality is considering the Jordan--Hahn decomposition of $W_0 - W_1$, which gives two orthogonal PSD matrices $W_+, W_-$ such that $\Tr\mparen{W_+} = \Tr\mparen{W_-} = \frac12\norm{W_0 - W_1}_1$.
This is because $\Tr\mparen{W_+} + \Tr\mparen{W_-} = \Tr\mparen{W_+ + W_-} = \norm{W_0 - W_1}_1$ and $\Tr\mparen{W_+} - \Tr\mparen{W_-} = \Tr\mparen{W_+ - W_-} = \Tr\mparen{W_0 - W_1} = 0$.

Then the overall trace distance of doing a random MUB measurement is, by Cauchy--Schwarz,
\begin{align*}
    \frac12\E_r\mbracket{\norm{W_0^{(r)} - W_1^{(r)}}_1} &\le \frac12 \sqrt N \cdot \E_r\mbracket{\norm{W_0^{(r)} - W_1^{(r)}}} \\
    &\le \frac{\sqrt N}2 \cdot \sqrt{\E_r\mbracket{\norm{W_0^{(r)} - W_1^{(r)}}^2}} \\
    &\le \sqrt{\frac N{2p}} \cdot \frac{\norm{W_0 - W_1}_1}2,
\end{align*}
where the second inequality is due to Jensen's inequality.
\end{proof}

In \iflncs{}the full version\else\Cref{appendix:alternative-bases}\fi, we further discuss a few alternative choices of bases of measurement that work or do not work.

\section{Commitments from Hard Classical\extrapolateto{}Quantum Extrapolation}

\subsection{Classical\extrapolateto{}Quantum Extrapolation}
\label{sec:cq-extrap}

Here, we define a computational task which we call \emph{classical$\rightarrow$quantum extrapolation}. At a high level, a classical$\rightarrow$quantum extrapolation problem is specified by an efficiently sampleable distribution over pairs $(s, \ket{\psi_s})$ consisting of a classical string $s$ and a quantum state $\ket{\psi_s}$. Given a random $s$, the adversary's task is to extrapolate the quantum half of the pair $\ket{\psi_s}$. We say that this task is hard if no QPT adversary succeeds at this with noticeable probability over the choice of $s$.

\begin{definition}%
    A cloneable$\rightarrow$quantum extrapolation problem is specified by two efficient families of isometries $(\mathsf{Gen}, \mathsf{Clone})$ satisfying the following property.
    \begin{itemize}
        \item (Cloning correctness) For every $\lambda$, let $\braces{\ket{\mathsf{chal}_i}}_i$ be the set of unit vectors such that $\mathsf{Clone}_\lambda \ket{\mathsf{chal}_i} = \ket{\mathsf{chal}_i}\ket{\mathsf{chal}_i}$.\footnote{By the fact that isometries preserve inner products, it is not hard to see that all vectors in this set are mutually orthogonal. However, this might not form a basis since the set might be empty.}
        Then the state generated by $\mathsf{Gen}_\lambda$ should admit the following decomposition $\mathsf{Gen}_\lambda = \sum_i \alpha_i \ket{\mathsf{chal}_i}_\RegA \otimes \ket{\psi_i}_\RegB$ for some $\alpha_i \ge 0$ and unit vector $\ket{\psi_i}$.
        In other words, $\paren{\paren{I - \sum_i\proj{\mathsf{chal}_i}} \otimes I}\mathsf{Gen}_\lambda = 0$.
    \end{itemize}
    The goal of the task is on input $\ket{\mathsf{chal}_i}$, produce $\ket{\psi_i}$ with non-negligible probability, or more formally, the probability that $\Adv$ with auxiliary input $\aux$ succeeds
    \[
        \Tr\mparen{\sum_i \alpha_i^2 \proj{\psi_i} \Adv\mparen{\proj{\mathsf{chal}_i} \otimes \aux}}
    \]
    is non-negligible.\footnote{Here, we do not require that the target states $\ket{\psi_i}$ are orthogonal. For example, in public key quantum money, different serial numbers may recognize overlapping banknotes. The challenger can prepare the challenge mixed state $\sum_i \alpha_i^2 \proj{\mathsf{chal}_i}$ in the adversary's view by preparing $\Gen_\lambda$, then cloning $\ket{\mathsf{chal}_i}$ into an internal register.}

    In particular, if $\mathsf{Clone}\ket x = \ket x\ket x$ is simply cloning in the computational basis, then we consider this as a \emph{classical$\rightarrow$quantum} extrapolation task.
\end{definition}

Here, we make the distinction between efficiently cloneable basis and classical (or telegraphable) basis since efficiently cloneable basis might be a larger class \cite{ITCS:NZ24}.

\subsection{Hardness from Cryptographic Primitives}

Hard classical$\rightarrow$quantum extrapolation is implied by many natural quantum primitives, such as public-key quantum money and signatures with classical verification keys.
Morally, such hard extrapolation task is simply the impossibility to sample the secret quantum state conditioned on the public classical verification key; the reduction simply formalizes this idea using standard techniques.
In slightly more details,
\begin{itemize}
    \item In the case of public-key quantum money (or even a mini-scheme), no QPT adversary who is given the public verification key $\mathsf{vk}$ should be able to construct a state $\ket{\$}$ which passes $\Verify(\mathsf{vk}, \ket{\$})$ with noticeable probability; otherwise, they could ``clone'' or forge banknotes just by looking at the verification key.
    Since it is efficent to sample $\mathsf{vk}$ together with a passing banknote $\ket{\$}$, any adversary for the classical$\rightarrow$quantum extrapolation task defined by that sampling procedure forges banknotes with noticeable probability.
    This also extends to public-key quantum money where the public key is only cloneable instead of classical (and we get a hard cloneable$\rightarrow$quantum extrapolation task instead), or if the money has a serial number.
    \item Similarly, we can construct a hard cloneable$\rightarrow$quantum extrapolation task from signatures with cloneable verification keys.
    The task is again ``inverting'' the quantum secret state corresponding to the cloneable verification key that is used to sign additional messages.
    If we can create the quantum secret state with some noticeable overlap, then the forged signature would also pass verification with some noticeable probability as well.
\end{itemize}

Finally, we consider a slightly different task of quantum key distribution (QKD) with a bounded number of rounds.
In QKD, two parties, Alice and Bob, have access to an insecure quantum channel and a classically authenticated channel\footnote{Having access to a classically authenticated channel is also required classically: it is easy to see that without any authenticated channel the task is trivially impossible since the interceptor could always pretend to be Bob to Alice and pretend to be Alice to Bob. Public key infrastructure also does not get around this problem since it still assumes authenticated channels between the parties and the authority.}.
In the protocol, they take turns to send messages over the two channels with Alice sending the first message.
The malicious party, called the interceptor, can passively monitor the messages in the classical authenticated channel and can arbitrarily modify messages in the quantum channel.

There are many variants of QKD security considered in the literature.
Here we consider a very weak security notion, which we note is a search-type assumption.

\begin{definition}
    Consider a QKD protocol.
    Let $k_A, k_B, k_*$ be the keys Alice, Bob, and the interceptor produce correspondingly at the end of the protocol.
    We require the QKD protocol to be correct, meaning that without the presence of the interceptor, $\Pr[k_A = k_B \neq \bot]$ is negligibly close to $1$.
    
    We say the QKD protocol is (statistically) unpredictably-secure if for all efficient (or unbounded, resp.)\ interceptors, $\Pr[k_* = k_A \neq \bot \lor k_* = k_B \neq \bot]$ is negligible.
\end{definition}

\begin{fact}[{\cite[Protocol 2]{PRL:PS00}}]
    There exists a 3-message QKD that is statistically unpredictably-secure.
    Furthermore, the only quantum message is sent in the first message.
\end{fact}
\begin{fact}[{\cite{CRYPTO:MW24}}]
    Assuming the existence of post-quantum one-way functions, there exists a 2-message QKD that is computationally unpredictably-secure.
\end{fact}
\begin{fact}
    1-message QKD cannot be unpredictably secure.
\end{fact}
\begin{proof}[Proof sketch]
    This follows from a straightforward observation that the interceptor could simply perform Bob's honest action coherently and do a gentle measurement to extract Bob's key.
    The measurement is gentle since Alice's key is already determined after the first message, and by correctness, Bob should output the correct key with overwhelming probability.
    Finally, note that this attack is efficient.
\end{proof}

Given these known results, it is natural to wonder if 2-message QKD can be statistically secure, and whether it requires computational assumptions.
We observe that hard classical$\rightarrow$quantum extrapolation task also follows from 2-message QKD and thus it requires computational assumptions somewhere between a post-quantum one-way function and a hard classical$\rightarrow$quantum extrapolation task.
In fact, we can prove something slightly stronger.

\begin{proposition}
    If there is an unpredictably-secure QKD where all but the last two messages are classical, then there exists a hard classical$\rightarrow$quantum extrapolation task.
\end{proposition}
\begin{proof}
Without loss of generality, let us say Alice receives the last message.
Then consider the state immediately after Alice sends her quantum message, and the task is to produce the joint quantum state on Alice's and Bob's internal register and the message register conditioned on the classical transcript so far: this is a classical$\rightarrow$quantum extrapolation task.
Note that this is the state before any quantum communication reaches the other party, thus Bob's internal registers must be unentangled from the rest conditioned on the classical transcript.

Suppose for contradiction that this is easy, then we construct the interceptor as follows:
\begin{itemize}
    \item Upon receipt of Alice's last classical message, disgard her quantum message and produce a fresh copy of the joint quantum state and use its message register as the quantum message to send to Bob.
    \item Upon receipt of Bob's messages, run Alice's honest action to extract Bob's key.
\end{itemize}
By correctness of the QKD protocol, whenever we succeed in producing the quantum state with some noticeable overlap (which happens with non-negligible probability by assumption), Bob will not abort the protocol and we can predict Bob's key with roughly the same probability, breaking its security.
\end{proof}

Note that as special cases, this proposition covers any 2-message QKD, or any 3-message QKD whose first message is classical.

\subsection{A Quantum Bit Commitment Scheme}

Let $\{\Rand_k\}_{k \in \{0,1\}^{\poly(n)}}$ be a family of efficiently implementable $n$-qubit ``randomizing'' unitaries, which satisfy
\begin{equation}
    \label{eq:weak-hiding}
        \E_{k\gets \{0,1\}^{\poly(n)}}\mbracket{ \TD\mparen{\proj{\psi_0}^{(k)}, \proj{\psi_1}^{(k)}}}
        \le
        \frac{1}{\sqrt{2}}
\end{equation}
for any orthogonal states $\ket{\psi_0}$ and $\ket{\psi_1}$, where $\proj{\psi}^{(k)}$ is mixed state resulting from applying $\Rand_k$ to $\ket{\psi}$ and measuring it in the standard basis.\footnote{The key may be much longer than the number of qubits operated on, for example if using random Cliffords.}
These exist by \Cref{lem:mub}\iflncs\else or \Cref{lem:2-design}\fi.
More explicitly, for $P_x^{k} \coloneqq \Rand_k\proj{x}\Rand_k^\dagger $, we define
\[
    \proj{\psi}^{(k)} \coloneqq \sum_{x\in \{0,1\}^n} P_x^k \proj{\psi} P_x^k.
\]
Roughly speaking, for any fixed state $\ket{\psi}$, with high probability over $k \gets \{0,1\}^n$, the state $\Rand_k \ket{\psi}$ is well-spread in the standard basis (i.e., measuring it gives an outcome that is \emph{somewhat} uniform). 

Let $\CRand$ denote the controlled $\Rand_k$ unitary 
\begin{align*}
    \CRand \coloneqq \sum_{k \in \{0,1\}^{\poly(n)}} \proj{k} \otimes \Rand_k.
\end{align*}
\noindent We use the notation $\CRand_{\RegA \rightarrow \RegB}$ as shorthand for $\sum_k \proj{k}_{\RegA} \otimes \left( \Rand_k\right)_{\RegB}$, and we follow a similar convention for other controlled gates such as $\CNOT$.

\begin{construction}[Weakly Hiding Commitment]
    \label{constr:weak-hiding} 
    Let $(\Gen, \Clone)$ be a cloneable$\rightarrow$quantum extrapolation task.
    Our weakly-binding quantum bit commitment works as follows. 
\begin{itemize}
    \item To commit to a bit $b$:
    \begin{enumerate}
    \item The sender first prepares the following initial state (which is independent of $b$):
    \begin{align*}
        \ket{\mathsf{init}} \coloneqq \ket{0}_{\RegC_0} \otimes \ket{+^n}_{\RegK} \otimes \left(\sum_s \beta_s \ket{\chal_s}_{\RegS} \ket{\psi_s}_{\RegM_0}\right) \otimes \ket{0}_{\RegM_1},
    \end{align*}
    where $\RegM_0,\RegM_1$ are both $m$ qubit registers, and $\RegC_0$ is the same as size $(\RegK,\RegM_0)$.
    \item Next apply $\bigl(\Clone_{\RegS \rightarrow (\RegS, \RegC_1)} \cdot \CNOT_{(\RegK,\RegM_b) \rightarrow \RegC_0}\bigr) \cdot \bigl(\CRand_{\RegK \rightarrow \RegM_b}\bigr)$.
    \item Finally, send the $\RegC = (\RegC_0, \RegC_1)$ register.
    \end{enumerate}
    \item To decommit, send $b$ along with all the remaining registers $\RegK,\RegS,\RegM_0,\RegM_1$. The receiver verifies by checking that the state is the expected state
    \begin{align*}
        \ket{\com_b} \coloneqq \bigl(\Clone_{\RegS \rightarrow (\RegS, \RegC_1)} \cdot \CNOT_{(\RegK,\RegM_b) \rightarrow \RegC_0}\bigr) \cdot \bigl(\CRand_{\RegK \rightarrow \RegM_b}\bigr) \cdot  \ket{\init}.
    \end{align*}
\end{itemize}
\end{construction}

\begin{theorem}
    If $(\Gen, \Clone)$ is a hard cloneable$\rightarrow$quantum extrapolation problem, then \Cref{constr:weak-hiding} is a $\frac{1}{\sqrt{2}}$-statistically hiding and computationally binding commitment.
\end{theorem}
\begin{proof}
    We prove $\frac{1}{\sqrt{2}}$-statistically hiding in \Cref{claim:weak-hiding} and prove computational binding in \Cref{claim:comp-binding}.
\end{proof}

We compile the weakly-statistically hiding construction to be statistically hiding using standard XOR amplification.
In detail:

\begin{construction}[Full Commitment]\label{constr:full-commitment} The commitment scheme works as follows.
\begin{itemize}
    \item To commit to a bit $b$:
    \begin{enumerate}
        \item The sender samples a string $x\gets\{0,1\}^\secp$ such that the parity of $x$ is $b$.
        \item For each bit $x_i$ of $x$, the sender commits to $x_i$ using \Cref{constr:weak-hiding}.
    \end{enumerate}
    \item To decommit, send $b$, $x$ and the openings of the commitments to each bit of $x$. The verifier verifies by checking that for each $i\in [n]$, commitment $i$ validly opens to $x_i$ according to \Cref{constr:weak-hiding}, then checks that the parity of $x$ is $b$.
\end{itemize}
\end{construction}

\begin{corollary}
     If $(\Gen, \Clone)$ is a hard cloneable$\rightarrow$quantum extrapolation problem, then \Cref{constr:full-commitment} is a statistically hiding and computationally binding commitment.
\end{corollary}
\begin{proof}
    Statistical hiding amplification follows from the fact that trace distance is amplified exponentially under XOR compositions \cite[Lemma 2]{Wat02-qszk}.
    Computational binding holds since in order to change the final bit revealed, the adversary must flip at least one of the $\lambda$ bits committed using the original commitment scheme; by randomly guessing the index, we complete the reduction to breaking the binding security of the original commitment scheme.
\end{proof}

\subsection{Statistical Hiding}

\begin{claim}\label{claim:weak-hiding}
    \Cref{constr:weak-hiding} is $\frac{1}{\sqrt{2}}$-statistically hiding.
\end{claim}
\begin{proof}
    The state of $(\RegC, \RegR)$ after a commitment to $b$ is:
    \begin{align*}
        \ket{\com_0} &\coloneqq \sum_{s,k,x} \beta_s 2^{-n/2} \ket{k, \chal_s, x}_{\RegC} \otimes \left(\ket{k,\chal_s} \otimes \proj{x} \Rand_k \ket{\psi_s} \otimes \ket{0} \right)_{\RegD},
        \\
        \ket{\com_1} &\coloneqq \sum_{s,k,x} \beta_s 2^{-n/2} \ket{k, \chal_s, x}_{\RegC} \otimes \left(\ket{k,\chal_s} \otimes \ket{\psi_s} \otimes \proj{x} \Rand_k \ket{0} \right)_{\RegD}.
    \end{align*}
    After tracing out register $\RegD$, the mixed states on register $\RegC$ are
    \begin{align*}
        \rho_0 &\coloneqq \sum_{s} \beta_s^2 \proj{\chal_s} \otimes  \sum_{k \in \{0,1\}^n} 2^{-n}\proj{k} \otimes \sum_{x \in \{0,1\}^n}   \left|\bra{x} \Rand_k \ket{\psi_s}\right|^2 \proj{x},
        \\
        \rho_1 &\coloneqq \sum_{s} \beta_s^2 \proj{\chal_s} \otimes \sum_{k \in \{0,1\}^n} 2^{-n}\proj{k} \otimes \sum_{x \in \{0,1\}^n}  \left|\bra{x} \Rand_k \ket{0}\right|^2 \proj{x}.
    \end{align*}
    Denote the projections of $\ket{\psi_s}$ and $\ket{0}$, respectively, onto measurement in the basis $\Rand_k$ as
    \begin{align*}
        \rho_{0,s,k} &\coloneqq \sum_{x \in \{0,1\}^n}   \left|\bra{x} \Rand_k \ket{\psi_s}\right|^2 \proj{x},
        \\
        \rho_{1,s,k} &\coloneqq \sum_{x \in \{0,1\}^n}   \left|\bra{x} \Rand_k \ket{0}\right|^2 \proj{x}.
    \end{align*}
    Then the trace distance between $\rho_0$ and $\rho_1$ is bounded as
    \begin{align}
        \frac{1}{2}\|\rho_0 - \rho_1\|_1 
        &\leq \frac{1}{2}\sum_{s} \beta_s^2  \sum_{k \in \{0,1\}^n} 2^{-n} \left\| \rho_{0,s,k} - \rho_{1,s,k} \right\|_1
        \label{eq:wh-convex}
        \\
        &= \sum_{s} \beta_s^2  \E_{k \in \{0,1\}^n}\mbracket{ \frac{1}{2}\left\| \rho_{0,s,k} - \rho_{1,s,k} \right\|_1}
        \nonumber
        \\
        &\leq \sum_{s} \beta_s^2 \frac{1}{\sqrt{2}}
        \label{eq:wh-lemmas}
        \\
        &= \frac{1}{\sqrt{2}}. \nonumber
    \end{align}
    \eqref{eq:wh-convex} follows from the convexity of the trace norm and its invariance under tensor product with the same state; \eqref{eq:wh-lemmas} follows from \eqref{eq:weak-hiding}.
\end{proof}

\subsection{Computational Binding}

\begin{claim}\label{claim:comp-binding}
    If $(\Gen, \Clone)$ is a hard cloneable$\rightarrow$quantum extrapolation problem, then \Cref{constr:weak-hiding} is computationally binding.
\end{claim}
\begin{proof}

Suppose we have an adversary $\Adv$ that can map $\ket{\com_0}$ to $\ket{\com_1}$ by acting only on the decommitment register $(\RegK,\RegS,\RegM_0,\RegM_1)$. 

We'll use this to give a solver for the cloneable$\rightarrow$quantum extrapolation problem. Let's think of the cloneable$\rightarrow$quantum extrapolation problem as follows. The challenger prepares 
\begin{align*}
    \ket{\chi_{\mathrm{start}}} \coloneqq \sum_s \beta_s \ket{\chal_s}_{\RegS} \ket{\psi_s}_{\RegM} \otimes \ket{\chal_s}_{\RegS'}
\end{align*} and gives us the $\RegS'$ register. To win the game, it suffices to turn this into the pure state
\begin{align*}
    \ket{\chi_{\mathrm{end}}}\coloneqq \sum_s \beta_s \ket{\chal_s}_{\RegS} \ket{\psi_s}_{\RegM} \otimes \ket{\aux_s}_{\RegA} \ket{\psi_s}_{\RegM'}
\end{align*}
for some $\aux_s$ since the $\RegM'$ register will also pass the verification.

The reduction, given the $(\RegS',\RegM')$ register of $\ket{\chi_{\mathrm{start}}}$, initializes a $\RegK$ register to the uniform superposition $\ket{+^n}$, $\RegM_1$ registers to $\ket{0}$, and then applies $\bigl(\Clone_{(\RegK,\RegS,\RegM') \rightarrow \RegC}\bigr) \cdot \bigl(\CRand_{\RegK \rightarrow \RegM'}\bigr)$. It applies the binding adversary on registers $\RegK, \RegS', \RegM', \RegM_1$.
It then coherently checks that the registers $\RegK, \RegS', \RegM_1$ agrees\footnote{This can be done efficiently for any cloneable basis. In particular, take any efficient unitary completion of $\Clone$ and we simply compute its inverse and check if the auxiliary registers return to $\ket0$.} with $\RegC$ and abort if not (in other words only the decision is measured).
Finally, it returns register $\RegM'$ to the challenger.

\paragraph{Random measurements.} In general, the result of applying $\Rand_k$ can be expressed as
\begin{align*}
    \Rand_k \ket{\psi_{s}} &= \sum_{x} \alpha_{k, s, x} \ket{x},
    \\
    \Rand_k \ket{0} &= \sum_{x} \alpha_{k, x} \ket{x}
\end{align*}
for some complex coefficient $\alpha$'s.
Using these, we can write the commitment states as
\begin{align*}
    \ket{\com_0} &= \sum_{k, s, x} 2^{-n/2} \beta_s \alpha_{k, s, x} \ket{k, \chal_s, x}_{\RegC} \otimes \ket{k, \chal_s, x, 0}_{\RegD},
    \\
    \ket{\com_1} &= \sum_{k, s, x} 2^{-n/2} \beta_s \alpha_{k, x} \ket{k, \chal_s, x}_{\RegC} \otimes \ket{k, \chal_s, \psi_s, x}_{\RegD}.
\end{align*}
Here, $\RegC$ contains the commitment (sent to the receiver) and $\RegD$ contains the opening (kept by the sender). 

\paragraph{Binding adversary's attack.} Say the adversary consists of a unitary $U_A$ and an auxiliary quantum input $\ket{\aux}_\RegA$. It applies $U_A$ to registers $\RegC$ and $\RegA$.
We assume the adversary always prepares a state $\ket{k,s,\phi_{k,s,x},x}$ in register $D$ along with potentially another private register.
This is without loss of generality since not preparing a state like this form can only decrease the binding attack advantage since when the challenger projects onto $\ket{\com_1}$ it necessarily projects $k, s, x$ to the correct value; furthermore, in our reduction, we also perform this check so the amplitude outside would not affect the reduction either.
More formally, we can account for this by considering the residual state to be potentially sub-normalized when the adversary does not comply.
The output of the adversary is
\[
    \ket{\varphi_A} = \sum_{k, s, x, y} 2^{-n/2} \beta_s \alpha_{k, s, x} \ket{k, \chal_s, x}_{\RegC} \otimes \ket{k, \chal_s, y, x}_{\RegD} \otimes \ket{\aux_{k,s,y,x}}_{\RegA}.
\]
The state $\ket{\phi_{k, s, x}} := \sum_y \ket{y}\ket{\aux_{k,s,y,x}}$ might be subnormalized but the remaining component will never contribute (in either direction) to the binding advantage.
From this we can see that the adversary's action can be seen as first swapping $\RegM_0 \RegM_1$ and then prepare a (possibly subnormalized) quantum state on the registers $\RegM_0\RegA$ controlled on $k, s, x$.

We now consider the state after projecting $\RegM_0$ to $\ket{\psi_s}$ to be
\[
    \ket{\varphi'_A} = \sum_{k, s, x} 2^{-n/2} \beta_s \alpha_{k, s, x} \ket{k, s, x}_{\RegC} \otimes \ket{k, \chal_s, \psi_s, x}_{\RegD} \otimes \ket{\aux_{k,s,x}}_{\RegA},
\]
where $\ket{\aux_{k,s,x}} := \sum_y \braket{\psi_s}{y}\ket{\aux_{k,s,y,x}}$.
Note that this projection again contains both the commitment receiver's projection as well as the extrapolation verification projection.
As a consequence, we can for simplicity consider the overlap with this state, which will be a lower bound on the success probability.

\paragraph{Additional notations.}
Let $\Pr[s] = \beta_s^2$ be the probability of sampling a challenge $s$.
Define $p_{k,s}(x) = |\alpha_{k,s,x}|^2$ to be the probability that a measurement of $\ket{\psi_{s}}$ in the basis $k$ outputs $x$. Similarly define $q_{k}(x) = |\alpha_{k,x}|^2$. $w(k,s,x)$ denotes the probability that the adversary, given $\ket{k,\chal_s,x}$, successfully produces $\ket{\psi_s}$. In other words,
\begin{align*}
     w(k,s,x) 
        &= \norm{\ket{\aux_{k,s,x}}}^2.
\end{align*}

\paragraph{Reduction success probability.} Given the notation above, the reduction's success probability in the cloneable$\rightarrow$quantum extrapolation game is
\[
    \sum_{k,s} 2^{-n} \Pr[s] \sum_{x} q_k(x) w(k,s,x) = \E_{k,s}\mbracket{\sum_{x} q_k(x) w(k,s,x)},
\]
since again, in the reduction we also trimmed out the useless amplitudes.

We now show that the reduction's success probability in the cloneable$\rightarrow$quantum extrapolation game is at least the adversary's probability of success in the binding attack.
To that end, we compute the binding success probability to be
\begin{align}
    \norm{(\bra{\com_1} \otimes I)\ket{\varphi'_A}}^2\nonumber
        &= \norm{2^{-n} \sum_{k,s,x}\beta_s^2  \alpha_{k,s,x} \overline{\alpha_{k,x}} \ket{\aux_{k,s,x}}}^2\nonumber
        \\
        &\le \paren{2^{-n} \sum_{k,s,x}\beta_s^2 \abs{ \alpha_{k,s,x} \overline{\alpha_{k,x}} } \cdot \norm{\ket{\aux_{k,s,x}}}}^2 \label{eq:binding-triangle}
        \\
        &= \E_{k, s}\mbracket{\sum_x\abs{ \alpha_{k,s,x} \alpha_{k,x}}\sqrt{w(k, s, x)}}^2 \nonumber \\
        &\le \E_{k, s}\mbracket{\mparen{\sum_x\abs{ \alpha_{k,s,x} \alpha_{k,x}}\sqrt{w(k, s, x)}}^2} \label{eq:binding-jensen} \\
        &\le \E_{k, s}\mbracket{\mparen{\sum_x\abs{ \alpha_{k,s,x}}^2} \mparen{\sum_x{\abs{\alpha_{k,x}}^2w(k, s, x)}}} \label{eq:binding-cs} \\
        &= \E_{k, s}\mbracket{\sum_{x} q_{k}(x) w(k,s,x)}.\nonumber
\end{align}
\eqref{eq:binding-triangle} follows from triangle inequality; \eqref{eq:binding-jensen} follows from Jensen's inequality; \eqref{eq:binding-cs} follows from Cauchy--Schwarz inequality.
Therefore, our reduction succeeds with the same probability as the binding adversary's advantage, completing the proof.
\end{proof}

\section{Quantum Extrapolation}

In this section, we define (fully) quantum extrapolation tasks and give some initial observations.

\begin{definition}[Quantum Extrapolation]
    \label{def:qextrapolation}
    Let $\Gen$ be a family of efficiently preparable bipartite pure states on registers $\RegA, \RegB$.
    Let $\rho$ be the reduced density matrix on $\RegA$.
    For every $\lambda$, define the target state to be its canonical purification
    \[
        \ket{T_\lambda} \coloneqq (\sqrt\rho \otimes I) \ket{\Phi}
    \]
    where $\ket\Phi_{\RegA\RegB'} = \sum_i \ket i\ket i$ is the unnormalized maximally entangled state where $\RegB'$ has same dimension as $\RegA$.
    $\Gen$ is a \emph{hard quantum extrapolation problem} if there exists a polynomial $p$ such that for every auxiliary input $\aux$ and polynomial-time channel $V_{\RegB\RegC \rightarrow \RegB'}$,
    \[
        F\mparen{\proj{T_\lambda}_{\RegA,\RegB'}, V_{\RegB\RegC \rightarrow \RegB'}(\Gen_\lambda() \otimes \aux_\lambda)} \leq 1 - \frac{1}{p(\lambda)}.
    \]
\end{definition}

Readers might notice that this is defined differently from \Cref{thm:informal-qextrapolation}.
We now explain why the two definitions are equivalent, or why \Cref{def:qextrapolation} can be viewed as an extrapolation problem.

\begin{lemma}
    Let
    \[
        \ket\Gen = \sum_{i} \alpha_i \ket{A_i}_{\RegA} \otimes \ket{B_i}_{\RegB}
    \]
    be any of $\ket\Gen$'s Schmidt decompositions.
    Then the target state
    \[ \ket{T} =  \sum_{i} \alpha_i \ket{A_i}_{\RegA} \otimes \ket{A_i^*}_{\RegB'} \]
    where $\ket{A_i^*}$ is the complex conjugate of $\ket{A_i}$.
\end{lemma}
\begin{proof}
    By the Schmidt decomposition, we can see that $\rho = \sum_i\alpha_i^2 \proj{A_i}$.
    Since $\{A_i\}_i$'s form an orthonormal basis, $\sqrt\rho = \sum_i\alpha_i \proj{A_i}$ and $\ket{\Phi} = \sum_i \ket{A_i} \otimes \ket{A_i^*}$.
    The lemma follows by plugging these into the definition of $\ket T$.
\end{proof}

\begin{corollary}
    For any $\Gen$, there exists an isometry $U$ such that $U_{\RegB\to\RegB'}\ket{\Gen} = \ket T$.
\end{corollary}
\begin{proof}
    Consider any Schmidt decomposition as above.
    Take any $U$ such that $U\ket{B_i} = \ket{A_i^*}$ for all $i$.
    This is well defined since $\ket{B_i}$'s and $\ket{A_i^*}$'s form a basis.
\end{proof}

The reader may notice that the Schmidt decomposition of a state is not necessarily unique.
Indeed, if we asked the adversary to map $\ket{B_i}$ to $\ket{A_i}$ directly, the optimal action would depend on the specific Schmidt decomposition that we consider.
However, by requiring the adversary to instead map it to $\ket{A_i^*}$ instead, we ensure uniqueness of the target state $\ket{T}$, and thus that the quantum extrapolation problem is well-defined.

Let us briefly compare this task with classical$\rightarrow$quantum extrapolation.
On a high level, there are three important differences:
\begin{enumerate}
    \item Obviously, the input or the challenge is now a quantum state instead of classical bitstrings.
    \item We also require the solver to solve the problem coherently, in other words, the solver cannot measure (or remember) which $i$ is given if he wishes to succeed.
      This is different from classical$\rightarrow$quantum extrapolation where the solver is free to measure the input.
      This condition is important to ensure the uniqueness of the action as otherwise the action would depend on the specific Schmidt decomposition that we consider.
    \item Finally, we also require the problem to be ``strongly'' solved, i.e.\ solving the task to arbitrary inverse polynomial error; whereas we can accept any non-negligible overlap for classical$\rightarrow$quantum extrapolation.\footnote{This difference would not be as important if one could show an amplification theorem for this task. Known techniques \cite{STOC:BQSY24} do not apply since it is inefficient to project onto the target state.}
\end{enumerate}

Before proceeding, let us first establish the robustness of the quantum extrapolation task.
Informally, we show that if the instance specified by $\Gen$ is slightly disturbed, then the target state will not change drastically either.
As a corollary, the same optimal action will still do pretty good even if a slightly disturbed input is given instead.
This is not obvious since the target state depends on $\Gen$ non-linearly.
In particular, the square-root superoperator is not linear.

To prove this, we first need the following lemma on Holevo fidelity $F_H(\rho, \sigma) := \Tr(\sqrt\rho \sqrt\sigma)$.

\begin{lemma}[{\cite{Kho72-fidelity}}]
    \label{conj:robustness}
    Let $\rho, \sigma$ be two density matrices.
    Then $1 - \sqrt{F_H(\rho, \sigma)} \le TD(\rho, \sigma) \le \sqrt{1 - F_H(\rho, \sigma)}$.
\end{lemma}

\begin{lemma}
    If two quantum extrapolation tasks $\abs{\braket{\Gen}{\Gen'}} = \sqrt{1 - \varepsilon}$ have a large overlap, then their target states has noticeable overlap as well: $\braket{T}{T'} \ge (1 - \sqrt\varepsilon)^2$.
\end{lemma}
\begin{proof}
    Let $\rho, \sigma$ be the reduced density matrices of $\Gen$ and $\Gen'$ respectively.
    By Uhlmann's theorem,
    \begin{equation}
        \label{eq:fidelity-from-uhlmann}
        \sqrt F(\rho, \sigma) \ge |\braket{\Gen}{\Gen'}| = \sqrt{1 - \varepsilon}.
    \end{equation}
    Since the target states are the canonical purifications of the bipartite states,
    we compute the overlap of the two target states
    \begin{align}
        \braket{T'}{T} &= \bra\Phi (\sqrt\sigma \otimes I) (\sqrt\rho \otimes I)\ket\Phi \nonumber \\
            &= \Tr(\sqrt\sigma \sqrt\rho) \nonumber \\
            &= F_H(\rho, \sigma) \nonumber \\
            &\ge (1 - TD(\rho, \sigma))^2 \label{eq:use-holevo} \\
            &\ge (1 - \sqrt{1 - F(\rho, \sigma)})^2 \label{eq:use-fuchs-vandegraaf} \\
            &\ge (1 - \sqrt{\varepsilon})^2 \label{eq:use-fidelity-from-uhlmann}.
    \end{align}
    \eqref{eq:use-holevo} is by \Cref{conj:robustness}; \eqref{eq:use-fuchs-vandegraaf} is by Fuchs--van de Graaf inequality; and finally, \eqref{eq:use-fidelity-from-uhlmann} is by \eqref{eq:fidelity-from-uhlmann}.
\end{proof}

\subsection{Relations to Other Hardness Assumptions}

\iflncs\else
In this section, we begin by showing that quantum extrapolation easily generalizes classical$\rightarrow$quantum extrapolation tasks, and then move onto construction from commitments.

\begin{proposition}\label{prop:pkqm-qextrap}
    If hard classical$\rightarrow$quantum extrapolation exists, then hard quantum extrapolation exists.
\end{proposition}
\begin{proof}
In classical$\rightarrow$quantum extrapolation, the challenger efficiently prepares the state
    \[
        U\ket{\vec{0}} = \sum_{x\in \{0,1\}^\lambda} \alpha_x \ket{A_x, x}_{\RegA} \otimes \ket{x}_{\RegB}
    \]
    and asks the adversary to find $\ket{A_x}$ given $x$.\footnote{Here, the extra copy of $x$ in $\RegA$ acts as a delayed measurement of $x$, which purifies the classical$\rightarrow$quantum extrapolation experiment.} Observe that the complex conjugate of this state can also be prepared efficiently by conjugating each gate in $U$, producing
    \begin{equation}
        \label{eq:cq2q-schmidt}
        U^*\ket{\vec{0}} = \left(U\ket{\vec{0}} \right)^*= \sum_{x\in \{0,1\}^\lambda} \alpha_x \ket{A_x^*, x}_{\RegA} \otimes \ket{x}_{\RegB}.
    \end{equation}
    Sice both $\ket{A_x^*, x}$'s as well as $\ket{x}$'s are orthonormal, \eqref{eq:cq2q-schmidt} must be a Schmidt decomposition.
    
    Assume for contradiction that quantum extrapolation were easy, then there exists a QPT adversary operating on register $\RegB$ of $U^*\ket{\vec{0}}$ which produces a state with $\frac12$ fidelity to
    \[
        \sum_{x\in \{0,1\}^\lambda} \alpha_x \ket{A_x^*, x}_{\RegA} \otimes \ket{A_x, x}_{\RegB}.
    \]
    We now show that this adversary, when applied to register $\RegB$ of $U\ket{\vec{0}}$, produces a state with noticeable fidelity to 
    \[
        \sum_{x\in \{0,1\}^\lambda} \alpha_x \ket{A_x,x}_{\RegA} \otimes \ket{A_x,x}_{\RegB}.
    \]
    Thus, this contradicts the hardness of classical$\rightarrow$quantum extrapolation.
    
    Since $\ket{A_x, x}$ is orthogonal to $\ket{A_{x'}, x'}$ for any $x\neq x'$, there exists a unitary $M$ mapping $\ket{A_x, x} \mapsto \ket{A_x^*, x}$. Applying $M$ to register $\RegA$ maps $U\ket{\vec{0}}$ to $U^*\ket{\vec{0}}$ 
    Then, since 
    \begin{align*}
        (I_\RegA \otimes V_{\RegB,\RegC}) \left(U\ket{\vec{0}} \otimes \ket{\aux}_{\RegC}\right)
        &=
        (M^\dagger_\RegA M_\RegA \otimes V_{\RegB,\RegC}) \left(U\ket{\vec{0}} \otimes \ket{\aux}_{\RegC}\right)
        \\
        &= (M^\dagger_\RegA \otimes I_{\RegB,\RegC})(I_\RegA \otimes V_{\RegB,\RegC}) (M_\RegA \otimes I_{\RegB,\RegC})\left(U\ket{\vec{0}} \otimes \ket{\aux}_{\RegC}\right)
        \\
        &= (M^\dagger_\RegA \otimes I_{\RegB,\RegC})(I_\RegA \otimes V_{\RegB,\RegC}) \left(U^*\ket{\vec{0}} \otimes \ket{\aux}_{\RegC}\right),
    \end{align*}
    the adversary produces a state with fidelity $\frac12$ to
    \[
        (M^\dagger_{\RegA} \otimes I_{\RegB,\RegC})\sum_{x\in \{0,1\}^\lambda} \alpha_x \ket{A_x^*, x}_{\RegA} \otimes \ket{A_x, x}_{\RegB} \otimes I_{\RegC}
        =
        \sum_{x\in \{0,1\}^\lambda} \alpha_x \ket{A_x, x}_{\RegA} \otimes \ket{A_x, x}_{\RegB} \otimes I_{\RegC}.
    \]
    Since this state is always accepted by the classical$\rightarrow$quantum extrapolation challenger, the reduction must succeed with probability at least $\frac12$ as well, which is non-negligible.
\end{proof}

We next prove that the hardness is also implied by commitments.
This technically subsumes the proposition above, but the reduction is slightly more involved.
\fi

\begin{lemma}[Triangle-Like Inequality for Fidelity \cite{NielsenChuang}]\label{lem:triangle-like-fidelity}
    For any three states $\rho$, $\sigma$, and $\tau$, the following inequality holds:
    \[
        \arccos(F(\rho, \tau)) \leq \arccos(F(\rho, \sigma)) + \arccos(F(\sigma, \tau)).
    \]
\end{lemma}

\begin{theorem}
    If quantum bit commitments exist, then hard quantum extrapolation exists.
\end{theorem}
\begin{proof}
    Without loss of generality, we can simply consider this for statistically hiding commitment schemes, which can be constructed from general quantum bit commitments \cite{AC:Yan22,ITCS:BCQ23}.
    Consider the commitment states $\ket{\Com_0}, \ket{\Com_1}$ when committing to $0, 1$ respectively.
    Let $\rho_b := \Tr_D(\proj{\Com_b})$ be the commitment message for bit $b$.
    By statistical hiding, we have that $TD(\rho_0, \rho_1) \le .01$ for all sufficiently large security parameters.

    Let the target states $\ket{T_b} = (\sqrt{\rho_b} \otimes I)\ket{\Phi}$.
    On a high level, we will construct an adversary that turns $\ket{\Com_0}$ to $\ket{T_0}$ using quantum extrapolation, and then pretend that $\ket{T_0} \approx \ket{T_1}$ by statistical hiding, and then undo the quantum extrapolation to turn $\ket{T_1}$ to $\ket{\Com_1}$, breaking the binding property.
    
    Assume for contradiction that quantum extrapolation were easy, then there exist efficient quantum channels $\cA_b$ that only act on the $\RegD$ register, such that $\Tr(\proj{T_b}\cA_b(\proj{\Com_b})) \ge .99$ for both $b = 0, 1$.
    Note that this also implies the ability to efficiently invert the extrapolation: there exists efficient $\cA'_b$ that only acts on the $\RegD$ register, such that $\Tr(\proj{\Com_b}\cA_b'(\proj{T_b})) \ge .99$ for both $b = 0, 1$.\footnote{One essentially runs the unitary purification of $\cA_b(\proj{\Com_b})$ for some sufficiently high-fidelity $\cA_b$, swap its output register with the input and then uncompute the purified $\cA_b$.}
    
    The adversary for breaking computational binding is simply $\cA'_1 \circ \cA_0$.
    We now analyze the success probability.
    After running $\cA_0$, we must have some state that has fidelity $\ge .99$ with $\ket{T_0}$.
    Since $\abs{\braket{T_1}{T_0}}^2 = \abs{\bra\Phi (\sqrt\rho_1 \otimes I) (\sqrt\rho_0 \otimes I)\ket\Phi}^2 = \abs{\Tr(\sqrt\rho_1 \sqrt\rho_0)}^2 = \abs{F_H(\rho_0, \rho_1)}^2 \ge (1 - TD(\rho_0, \rho_1))^4 \ge .96$ by \Cref{conj:robustness}, this state must have 0fidelity $\ge .9$ with $\ket{T_1}$ by \Cref{lem:triangle-like-fidelity}.
    Using \Cref{lem:triangle-like-fidelity} again, we get that after running $\cA'_1$, the output state has fidelity $\ge .8$ with $\ket{\Com_1}$.
    Therefore, we break computational binding with noticeable probability for all sufficiently large security parameters, a contradiction.
    \end{proof}

\subsection{Polynomial-Space Upper Bound}

Finally, we give an upper bound on how one could solve any quantum extrapolation task in (quantum) polynomial space, more specifically $\mathsf{pureUnitaryPSPACE}$ \cite[Definition 2.7]{MY23-stateQIP}.
This suggests that it might be hard to unconditionally establish the hardness of a quantum extrapolation task with an explicit classical description.

\begin{theorem}
    For any quantum extrapolation task specified by $\Gen$, it can be solved in quantum polynomial space and exponential time.
\end{theorem}
\begin{proof}[Proof sketch]
    The algorithm essentially reduces to a special case of a succinct Uhlmann instance, which can be solved in quantum polynomial space \cite{BEMPQY23}.
    To solve a quantum extrapolation task specified by $\Gen$, it is equivalent to consider the fidelity-1 Uhlmann transformation problem from $\ket{\Gen}$ to $\ket T$, or equivalently put, consider breaking the binding of a (not necessarily efficient) commitment scheme where the two commitment states are $\ket{\Gen}$ and $\ket T$.
    Recall that the target state is simply the canonical purification of $\ket{\Gen}$, i.e.\ consider the reduced density matrix $\rho := \Tr_B\mparen{\proj\Gen}$, then $\ket T = (\sqrt\rho \otimes I) \ket\Phi$ where $\ket\Phi_{\RegA\RegB}$ is the unnormalized maximally entangled state.

    We prove this by invoking Algorithmic Uhlmann's Theorem \cite[Theorem 7.4]{MY23-stateQIP}, which states that we can (approximately) perform the Uhlmann transformation problem in polynomial space given that we know how to (approximately) prepare each of the two states in polynomial space.
    Therefore, it suffices to show how to prepare the state $\ket T$ in polynomial space, or equivalently, to compute the amplitudes of $\ket T$ in $\mathsf{PSPACE}$; and this follows from the fact that we can space-efficiently approximate the block encoding of $\rho$ \cite[Lemma 3.3]{MY23-stateQIP}, the square root superoperator \cite[Lemma 3.15]{MY23-stateQIP}, as well as the product of two operators \cite[Lemma 3.5]{MY23-stateQIP}, and from there it is easy to compute the amplitudes from the block encoding.
\end{proof}

\ifanonymous\else
\section*{Acknowledgements}

We thank Fermi Ma for his extended insightful discussions.
We also thank John Bostanci, Yuval Efron, Tony Metger, Alexander Poremba, and Henry Yuen for their initial discussions on the quantum extrapolation task.
Finally, we thank Tomoyuki Morimae for the reference on the 3-message QKD protocol \cite{PRL:PS00}.

This work was done in part when LQ was at Boston University and supported by DARPA under Agreement No.\ HR00112020023.
\fi
\fi

\printbibliography

\ifexabs\else
\iflncs\else
\appendix

\section{Alternative choices of bases}
\label{appendix:alternative-bases}

\subsection{A Counterexample for Binary Phase Bases}

Without loss of generality, we are going to consider a basis to be described by a unitary followed by a complete measurement in the standard basis.
In this language, a set of bases can be equivalently expressed as a set of unitaries.

One natural choice of basis, inspired by the binary phase construction of pseudorandom states \cite{C:JiLiuSon18,TCC:BraShm19,TCC:AGQY22}, is to first apply a (pseudo-)random $\pm1$ diagonal matrix and then measure in the Hadamard basis.
Consider the counterexample where we have two states $\ket{\pm 0 \cdots 0} = \frac1{\sqrt2} (\ket{0^n} \pm \ket{10^{n-1}})$.
(In other words, state $b$ is $\frac1{\sqrt2} (\ket{0^n} + (-1)^b \ket{10^{n-1}})$.)
After applying a random diagonal matrix, these states are still either $\ket{\pm 0 \cdots 0}$ or $\ket{\mp 0 \cdots 0}$, and thus perfectly distinguishable when measured in the Hadamard basis.

\subsection{A Counterexample for Pauli Bases}

Consider two states $\ket{0^n}$ and $\ket{1^n}$.
Applying a random Pauli and then measuring in the standard basis will not work in this case since with overwhelming probability at least $n/2 - o(n)$ qubits will be measured in the standard basis, thus the measurement outcome can be distinguished with overwhelming probability.
This shows that 1-design cannot work since Pauli group forms a 1-design.

\subsection{Statistical Hiding with $2$-Designs}\label{sec:t-design}

In this section, we are going to show that using a unitary 2-design suffices, and we can efficiently sample a unitary 2-design over $n$ qubits uniformly from $2^{5n} - 2^{3n}$ elements and it can be implemented in quasi-linear time using Clifford gates \cite{CLLW16}.

\begin{lemma}\label{lem:2-design}
    Let $D_2$ be a unitary 2-design over dimension $N$ and $D_\infty$ be the Haar random unitary distribution.
    For any two mixed states $\rho_0$ and $\rho_1$, the expected trace distance after a random measurement drawn from $D_2$ is
    \[
        \E_{U\sim D_2} \mbracket{\TD\mparen{\rho_0^{(U)}, \rho_1^{(U)}}} \leq  \sqrt{\frac N{2(N+1)}} \cdot \TD(\rho_0, \rho_1)
    \]
    where $\rho^{(U)} \coloneqq \sum_{x=0}^{N-1} \proj{x} U \rho U^\dagger \proj{x}$ is the projection onto the basis $\{U\ket{x}\}_{x}$.
\end{lemma}

\begin{proof}
Consider the normalized Jordan--Hahn decomposition of $\rho_0 - \rho_1 = TD(\rho_0, \rho_1) \cdot (P_+ - P_-)$ for two orthogonal density matrices $P_+, P_-$.
Then we compute
\begin{align*}
    &\equad \E_{U \sim D_2}\mbracket{TD\mparen{\sum_x{\proj x U\rho_0U^\dagger\proj x}, \sum_x{\proj x U\rho_1U^\dagger\proj x}}} \\
    &= \E_{U \sim D_2}\mbracket{\frac12 \sum_x\abs{\bra x U\rho_0U^\dagger\ket x - \bra x U\rho_1U^\dagger\ket x}} \\
    &= \E_{U \sim D_2}\mbracket{\frac12 \sum_x\abs{\bra x U(\rho_0 - \rho_1)U^\dagger\ket x}} \\
    &= \frac N2 TD(\rho_0, \rho_1) \cdot \E_{U \sim D_2}\mbracket{\abs{\bra 0 U(P_+ - P_-)U^\dagger\ket 0}} \\
    &\le \frac N2 TD(\rho_0, \rho_1) \cdot \sqrt{\E_{U \sim D_2}\mbracket{\paren{\bra 0 U(P_+ - P_-)U^\dagger\ket 0}^2}} \\
    &= \frac N2 TD(\rho_0, \rho_1) \cdot \sqrt{\E_{U \sim D_\infty}\mbracket{\paren{\bra 0 U(P_+ - P_-)U^\dagger\ket 0}^2}}.
\end{align*}
where the inequality is due to Jensen's inequality.

It remains to calculate the last term.
To that end, let $P_+ - P_- = \sum_i p_i\proj{\psi_i}$ be its eigendecomposition since it is Hermitian.
Plugging this in, we get that
\begin{align}
    &\equad \E_{U \sim D_\infty}\mbracket{\paren{\bra 0 U(P_+ - P_-)U^\dagger\ket 0}^2} \nonumber \\
    &= \E_{U \sim D_\infty}\mbracket{\paren{\bra 0 U\paren{\sum_i p_i\proj{\psi_i}}U^\dagger\ket 0}^2} \nonumber \\
    &= \E_{U \sim D_\infty}\mbracket{\sum_i p_i^2 \paren{\bra 0 U\proj{\psi_i}U^\dagger\ket 0}^2 + \sum_{i \neq j} p_ip_j \bra 0 U\proj{\psi_i}U^\dagger\proj0 U\proj{\psi_j}U^\dagger\ket 0} \nonumber \\
    &= \sum_i p_i^2 \cdot \frac2{N(N+1)} + \sum_{i \neq j} p_ip_j \cdot \frac1{N(N+1)} \label{eq:haar-state-calc} \\
    &= \frac1{N(N+1)} \cdot \sum_i\paren{p_i^2 + p_i\sum_jp_j} \nonumber \\
    &= \frac1{N(N+1)} \cdot \sum_ip_i^2 \label{eq:pj-balanced} \\
    &\le \frac2{N(N+1)}. \label{eq:norm-is-stat-dist}
\end{align}
\eqref{eq:haar-state-calc} holds since by unitary invariance, it is equivalent to consider $\bra0U\ket{\psi_i}$ as the overlap between a Haar random state $\ket\psi$ and $\ket0$, thus $\E_{U \sim D_\infty}\mbracket{\paren{\bra 0 U\proj{\psi_i}U^\dagger\ket 0}^2} = \E[\abs{\braket{0}{\psi}}^4]$, where $\ket\psi$ is a Haar random state; on the other hand, $\E[\abs{\braket{0}{\psi}}^4] = \frac2{N(N+1)}$ can be computed via the probability of measuring $\ket0\otimes\ket0$ on the maximally mixed state over the symmetric subspace $\textrm{Sym}_2^N$, which has dimension $\binom{N + 1}2$.
Similarly, $\E_{U \sim D_\infty}\mbracket{\bra 0 U\proj{\psi_i}U^\dagger\proj0 U\proj{\psi_j}U^\dagger\ket 0} = \E\mbracket{\abs{\braket{0}{\psi}\braket{1}{\psi}}^2} = \frac1{N(N+1)}$.
\eqref{eq:pj-balanced} holds since $\sum_jp_j = \Tr(P_+ - P_-) = 0$.
\eqref{eq:norm-is-stat-dist} holds since $\sum_ip_i^2 = \norm{P_+ - P_-}_2^2 = \norm{P_+}_2^2 + \norm{P_-}_2^2 \le \norm{P_+}_1^2 + \norm{P_-}_1^2 = 2$.

Putting everything together, we obtain that it is upper bounded by $TD(\rho_0, \rho_1) \cdot \sqrt{\frac N{2(N+1)}}$.
\end{proof}

\paragraph{Higher $t$-Designs.} Interestingly, the bound gets only a bit better than $1/\sqrt2$ if we use Haar random unitaries.
Consider $\rho_0 = \proj0$ and $\rho_1 = \proj1$.
The calculation starts the same as the last section except that we start with using $D_\infty$ instead of $D_2$ and we avoid the step involving Jensen's inequality, so we obtain that the trace distance is
\begin{align*}
    \frac N2 \cdot \E\mbracket{\abs{\abs{\braket{0}{\psi}}^2 - \abs{\braket{1}{\psi}}^2}}.
\end{align*}
We note that when $N \to \infty$, each $\abs{\braket{x}{\psi}}^2$ is approximately an independent draw from the exponential distribution with rate $1/N$ (also known as the Porter--Thomas distribution).
Let $X, Y$ be two such independent random variables with PDF $p$, then we get that it is approximately
\begin{align*}
    \frac N2 \cdot \E\mbracket{\abs{X - Y}}
      &= N \cdot \E_{X \ge Y}\mbracket{X - Y} \\
      &= N \cdot \int_0^\infty p(y)\int_y^\infty p(x)(x - y)dxdy \\
      &= \frac12.
\end{align*}
This implies that going to higher $t$-designs does not seem to lead to a more asymptotically efficient construction.
\fi
\fi

\end{document}